\documentclass{article}

\usepackage[preprint]{neurips_2024}

\usepackage[utf8]{inputenc} 
\usepackage[T1]{fontenc}    
\usepackage{hyperref}       
\usepackage{url}            
\usepackage{booktabs}       
\usepackage{amsfonts}       
\usepackage{nicefrac}       
\usepackage{microtype}      
\usepackage{xcolor}         
\usepackage{hyperref}
\usepackage{amssymb,amsthm,amsfonts,amscd,mathrsfs}
\usepackage{amsmath}
\usepackage{cleveref}
\usepackage{thm-restate}
\usepackage{natbib}  

\newtheorem{thm}{Theorem}

\title{Quantum Equilibrium Propagation: \\
Gradient-Descent Training of Quantum Systems
}

\author{
Benjamin Scellier\\
Rain AI\\
\texttt{benjamin@rain.ai} \\
}

\begin{document}

\maketitle

\begin{abstract}
Equilibrium propagation (EP) is a training framework for energy-based systems, i.e. systems whose physics minimizes an energy function. EP has been explored in various classical physical systems such as resistor networks, elastic networks, the classical Ising model and coupled phase oscillators. A key advantage of EP is that it achieves gradient descent on a cost function using the physics of the system to extract the weight gradients, making it a candidate for the development of energy-efficient processors for machine learning. We extend EP to quantum systems, where the energy function that is minimized is the mean energy functional (expectation value of the Hamiltonian), whose minimum is the ground state of the Hamiltonian. As examples, we study the settings of the transverse-field Ising model and the quantum harmonic oscillator network -- quantum analogues of the Ising model and elastic network.
\end{abstract}

\section{Introduction}

Machine learning (ML) has been predominantly powered by classical digital computing. Meanwhile, quantum computing and neuromorphic computing are explored as alternative paradigms to enhance ML capabilities. Quantum computing aims to leverage the principles of quantum mechanics, such as superposition, to encode and process information in ways that classical computers cannot, potentially handling exponentially larger amounts of information. In contrast, neuromorphic computing, taking inspiration from the brain's energy efficiency, aims to leverage analog physics and compute-in-memory platforms to significantly reduce the cost of inference and training in ML \citep{markovic2020physics}. An emerging field of research known as `physical learning' \citep{stern2023learning} shares similar goals with neuromorphic computing, but explores the inherent physics of any physical system for computation, without necessarily mimicking neurons and synapses.

A key lesson from ML research over the past decades is the effectiveness of frameworks (such as backpropagation) for optimizing cost functions. One challenge for neuromorphic computing and physical learning has been the search for frameworks that optimize cost functions while adhering to local computation and local learning rules, the latter being essential for implementation on analog compute-in-memory platforms. In recent years, several gradient-descent training frameworks for physical systems have been proposed. For instance, \citet{lopez2023self} introduced a framework applicable to arbitrary time-reversal invariant Hamiltonian systems, and \citet{wanjura2023fully} developed a method for extracting weight gradients in optical systems based on linear wave scattering. This paper focuses on the training framework known as equilibrium propagation.

Equilibrium propagation (EP), introduced in \citet{scellier2017equilibrium}, is a framework for energy-based systems, where physics drives the system's state towards the minimum of an energy function (equilibrium or steady state). EP extracts the gradients of the cost function using two equilibrium states corresponding to different boundary conditions, which are then used to locally adjust the trainable weights of the system. EP has been applied to various systems, including resistor networks \citet{kendall2020training}, elastic and flow networks \citep{stern2021supervised}, spiking networks \citep{martin2021eqspike}, the (classical) Ising model \citep{laydevant2024training}, and coupled phase oscillators \citep{wang2024training}. Recent experimental demonstrations have shown the applicability of EP on hardware: \citet{dillavou2022demonstration,dillavou2023machine} built two generations of self-learning resistor networks, \citet{altman2023experimental} built a self-learning elastic network, \citep{yi2023activity} used a variant of EP in a memristor crossbar array, and \citet{laydevant2024training} used EP on D-wave to train a classical Ising network (where, interestingly, they used quantum annealing to reach the ground state). Simulations have further underscored the potential of EP for ML applications: in particular, \citet{laborieux2022holomorphic} trained an energy-based convolutional network to classify a downsampled version of the ImageNet dataset. More broadly, \citep{zucchet2022beyond} have highlighted EP's general applicability to any bilevel optimization problem (beyond the training of energy-based systems), including meta-learning \citep{zucchet2022contrastive}.

We introduce Quantum Equilibrium Propagation (QEP), an extension of EP to quantum systems. In QEP, the system is brought to the ground state of its Hamiltonian, parameterized by real-valued trainable weights, to produce a prediction. The algorithm performs gradient descent on the expectation value of an observable, which serves as the cost function to optimize. Thus, in QEP, the classical EP's energy function is replaced by the system's Hamiltonian, and the equilibrium state extremizing the energy function is replaced by the ground state of the Hamiltonian. The central ingredient for translating from EP to QEP is the energy expectation value, minimized (more generally, extremized) at the Hamiltonian's ground state (more generally, eigenstates). Similar to EP, an interesting feature of QEP is the locality of the learning rule, which might be useful for the development of specialized quantum hardware with reduced classical overhead, where measurements of the weight gradients and adjustments of the trainable weights would be performed locally.

First we review EP, as well as the (classical) Ising model and elastic network model where EP has been used (Section~\ref{sec:equilibrium-propagation}). Then we present the concepts of quantum mechanics used in the context of QEP (Section~\ref{sec:quantum-mechanics}). Finally, we introduce QEP and study the settings of the transverse-field Ising network and the quantum harmonic oscillator network -- quantum analogues of the Ising model and elastic network (Section~\ref{quantum-equilibrium-propagation}).

We note that, during the writing of this manuscript, \citet{massar2024equilibrium} similarly extended EP to quantum systems via the variational formulation of the Hamiltonian's eigenstates as critical points of the mean energy functional. Additionally, \citet{massar2024equilibrium} studied the thermal case, where a thermodynamic system settles to the minimum of the free energy functional, and showed how to extract the weight gradients from thermal fluctuations alone.

\section{Equilibrium propagation}
\label{sec:equilibrium-propagation}

This section reviews the Equilibrium Propagation (EP) training framework \citep{scellier2017equilibrium}.

EP applies in energy-based systems, which are governed by dynamics that drive their state $s$ towards the minimum of an energy function $\mathcal{E}$. These systems contain trainable weights $w=(w_1,w_2,\ldots,w_M)$ and can take an input $x$ supplied as a boundary condition. We denote the corresponding energy function as $\mathcal{E}(w,x,s)$. During inference, given an input $x$, the system evolves towards its equilibrium or steady state, characterized by
\begin{equation}
\label{eq:equilibrium-state}
s(w,x) = \underset{s}{\arg\min} \; \mathcal{E}(w,x,s).
\end{equation}
The system thus implements a function $x \mapsto s(w,x)$, and training consists in adjusting the weights $w$ so that $s(w,\cdot)$ matches a desired input-output function. Similar to traditional ML paradigms, EP uses a cost function $\mathcal{C}(s(w,x),y)$ which, given an input $x$ and its associated desired output $y$, measures the accuracy of the prediction $s(w,x)$ by comparing it with $y$. Training the system can be formulated as a bilevel optimization problem:
\begin{align}
\text{minimize} \; & \mathcal{J}(w) = \mathbb{E}_{(x,y)} \left[ \mathcal{C}(s(w,x),y) \right], \\ 
\text{subject to} \; & s(w,x) = \underset{s}{\arg\min} \; \mathcal{E}(w,x,s),
\end{align}
The conventional method to solve this problem is gradient descent on the upper-level objective: at each step of training, an input-output pair $(x,y)$ is picked from the training data, and the trainable weights are adjusted in proportion to the gradient of the cost function:
\begin{equation}
\Delta w = - \eta \nabla_w \mathcal{C}(s(w,x),y).
\end{equation}
The remaining task is to obtain or estimate the weight gradients $\nabla_w \mathcal{C}(s(w,x),y)$.

The main advantage of EP is that it extracts these weight gradients using the system's physics. The central idea is to view the cost function $\mathcal{C}(s,y)$ as the energy of an interaction between the state variables ($s$) and desired outputs ($y$), which can be incorporated into the energy function $\mathcal{E}$ of the system to form the `total energy function',
\begin{equation}
\label{eq:total-energy-function}
\mathcal{E}^\beta(w,x,s,y) = \mathcal{E}(w,x,s) + \beta \mathcal{C}(s,y),
\end{equation}
where $\beta \in \mathbb{R}$ is a parameter termed the `nudging parameter' that controls the strength of this new interaction. EP proceeds in three steps:
\begin{enumerate}
\item Set $\beta=0$ and let the system settle to an equilibrium state $s_\star^0$ of $\mathcal{E}^0$, called the `free state', characterized by
\begin{equation}
s_\star^0 = \underset{s}{\arg\min} \; \mathcal{E}^0(w,x,s,y) = s(w,x).
\end{equation}
For each $k \in \{ 1,2,\cdots,M \}$, measure $\frac{\partial \mathcal{E}^0}{\partial w_k}(w_1,\ldots,w_M,x,s_\star^0,y)$, i.e. the partial derivative of the energy function with respect to $w_k$.
\item Set $\beta>0$ and let the system reach a new equilibrium state $s_\star^\beta$ of $\mathcal{E}^\beta$, called the `nudge state', characterized by
\begin{equation}
\label{eq:nudge-state}
s_\star^\beta = \underset{s}{\arg\min} \; \mathcal{E}^\beta(w,x,s,y).
\end{equation}
Measure again $\frac{\partial \mathcal{E}^\beta}{\partial w_k}(w_1,\ldots,w_M,x,s_\star^\beta,y)$ for each $k \in \{ 1,2,\cdots,M \}$, at the nudge state this time.
\item Update the trainable weights $w_1,w_2,\ldots,w_M$ as
\begin{equation}
\label{eq:learning-rule-ep}
\Delta w_k = \frac{\eta}{\beta} \left[ \frac{\partial \mathcal{E}^0}{\partial w_k}(w,x,s_\star^0,y) - \frac{\partial \mathcal{E}^\beta}{\partial w_k}(w,x,s_\star^\beta,y) \right],
\end{equation}
where $\eta>0$ is a small learning rate.
\end{enumerate}

The main theoretical result of EP is that the above contrastive learning rule \eqref{eq:learning-rule-ep} approximates one step of gradient descent on the cost function.

\begin{thm}[Equilibrium Propagation]
\label{thm:equilibrium-propagation}
The gradient of the cost function with respect to the trainable weights can be approximated as
\begin{align}
\label{eq:ep-formula}
\nabla_w \mathcal{C}(s(w,x),y) & = \left. \frac{d}{d\beta} \frac{\partial \mathcal{E}^\beta}{\partial w}(w,x,s_\star^\beta,y) \right|_{\beta=0} \\
& \approx \frac{1}{\beta} \left[ \frac{\partial \mathcal{E}^\beta}{\partial w}(w,x,s_\star^\beta,y) - \frac{\partial \mathcal{E}^0}{\partial w}(w,x,s_\star^0,y) \right].
\label{eq:error-term}
\end{align}
\end{thm}

Theorem~\ref{thm:equilibrium-propagation} is proved in \citet{scellier2017equilibrium}. Here are several important points.

\paragraph{Local learning rule.} An important feature of the EP framework is that the contrastive learning rule of Eq.~\eqref{eq:learning-rule-ep} is local if the following condition is met. Let $w = (w_1, w_2, \ldots, w_M)$ be the set of trainable weights, and assume that the energy function is separable, $\mathcal{E} = \mathcal{E}_1 + \mathcal{E}_2 + \cdots + \mathcal{E}_M$, where each $\mathcal{E}_k$ is the energy term of an interaction parameterized by $w_k$ (and $w_k$ only). Then the energy derivatives arising in the learning rule simplify as $\frac{\partial \mathcal{E}}{\partial w_k} = \frac{\partial \mathcal{E}_k}{\partial w_k}$. If the energy term $\mathcal{E}_k$ involves only state variables spatially close to $w_k$, the learning rule for $w_k$ is local in space. Various physical systems where EP has been applied satisfy this property \citep{kendall2020training,stern2021supervised,laydevant2024training,wang2024training}. Below we illustrate this property in the (classical) Ising model and elastic network model.

\paragraph{Improved gradient estimator.} 
The error term due to the finite difference approximation in Eq.~\eqref{eq:error-term} is of order $O(\beta)$. \citet{laborieux2021scaling} introduced the symmetric finite difference gradient estimator,
\begin{equation}
\label{eq:learning-rule-cep}
\nabla_w \mathcal{C}(s(w,x),y) \approx \frac{\eta}{2 \beta} \left[ \frac{\partial \mathcal{E}^\beta}{\partial w}(w,x,s_\star^\beta,y) - \frac{\partial \mathcal{E}^{-\beta}}{\partial w}(w,x,s_\star^{-\beta},y) \right],
\end{equation}
which reduces the error term to $O(\beta^2)$ and performs much better in practice. \citet{scellier2023energy} later pointed out that the contrastive learning rules of Eq.~\eqref{eq:learning-rule-ep} and Eq.~\eqref{eq:learning-rule-cep} perform (exact) gradient descent on contrastive functions, even far from the infinitesimal nudging regime $|\beta| \ll 1$, and that these contrastive functions have useful properties. In the case of Eq.~\eqref{eq:learning-rule-ep}, the contrastive function is
\begin{equation}
\mathcal{L}_\beta(w,x,y) = \frac{1}{\beta} \left[ \mathcal{E}^\beta(w,x,s_\star^\beta,y) - \mathcal{E}^0(w,x,s_\star^0,y) \right],
\end{equation}
which approximates the true cost function when $\beta \to 0$ and satisfies
\begin{equation}
\mathcal{L}_\beta(w,x,y) \leq \mathcal{C}(s(w,x),y) \leq \mathcal{L}_{-\beta}(w,x,y), \qquad \beta > 0.
\end{equation}
Namely, $\mathcal{L}_\beta(w,x,y)$ is a lower bound of $\mathcal{C}(s(w,x),y)$ when $\beta > 0$, and an upper bound when $\beta < 0$. Hence, it was also discovered that the learning rule of Eq.~\eqref{eq:learning-rule-ep} yields significantly better results with $\beta < 0$ than with $\beta > 0$.

\paragraph{Unstable equilibrium.} 
While we have assumed here that the equilibrium states $s(w,x)$ and $s_\star^\beta$ of Eq.~\eqref{eq:equilibrium-state} and Eq.~\eqref{eq:nudge-state} are stable (i.e. minima of their respective energy functions), Theorem~\ref{thm:equilibrium-propagation} is also valid when $s(w,x)$ and $s_\star^\beta$ are critical points (or saddle points), where the stationary conditions
\begin{equation}
\label{eq:stationary-state}
\frac{\partial \mathcal{E}}{\partial s}(w,x,s(w,x)) = 0, \qquad 
\frac{\partial \mathcal{E}^\beta}{\partial s}(w,x,s_\star^\beta,y) = 0
\end{equation}
are met -- see \citet[Chapter 2]{scellier2021deep} for a brief discussion. A condition for Theorem~\ref{thm:equilibrium-propagation} to hold is that the nudge state $s_\star^\beta$ must be the stationary state obtained as a smooth deformation of $s_\star^0$ as we gradually vary the nudging parameter from $0$ to $\beta \neq 0$.

\subsection{Ising model}
\label{sec:ising-model}

The (classical) Ising model of coupled spins is a widely studied model in physics. It has been explored as a computing platform for machine learning, and recently studied in the context of EP \citep{laydevant2024training}. The model consists of $N$ classical spins, characterized by their state $\sigma_k \in \{ +1, -1 \}$, representing ``up'' or ``down'' states. The state of the system is represented by the $N$-dimensional vector of spin states, $s = (\sigma_1, \sigma_2, \ldots, \sigma_N)$, and the Ising energy function is defined as
\begin{equation}
\label{eq:ising-energy-function}
\mathcal{E}_{\rm Ising}(\sigma_1, \sigma_2, \ldots, \sigma_N) = - \sum_{1 \leq j < k \leq N} J_{jk} \sigma_j \sigma_k - \sum_{k=1}^N h_k \sigma_k,
\end{equation}
where $J_{jk}$ represents the couplings between spins, and $h_k$ represents the bias fields applied to individual spins. These parameters serve as trainable weights in the model. The partial derivatives of the energy function with respect to these trainable weights are given by:
\begin{equation}
\frac{\partial \mathcal{E}_{\rm Ising}}{\partial J_{jk}} = - \sigma_j \sigma_k, \qquad \frac{\partial \mathcal{E}_{\rm Ising}}{\partial h_k} = - \sigma_k.
\end{equation}
\citet{laydevant2024training} implemented an Ising network on the D-Wave Ising machine, and trained it to classify the MNIST handwritten digits using EP. They employed the quantum annealing procedure of D-Wave to reach the ground state. They also emulated a small convolutional Ising network, using the Chimera architecture of D-Wave's chips to implement the necessary convolutional operations.

\subsection{Elastic network}
\label{sec:elastic-network}

The elastic network model, first studied by \citet{stern2021supervised} in the context of EP, consists of $N$ masses $m_1, m_2, \ldots, m_N$  interconnected by springs. Denoting the position of mass $m_i$ as $\vec{r}_i$, the elastic energy stored in the spring connecting $m_i$ to $m_j$ is given by $\frac{1}{2} k_{ij} \left( \| \vec{r}_i - \vec{r}_j \| - \ell_{ij} \right)^2$, where $k_{ij}$ is the spring constant and $\ell_j$ is the spring's rest length. The state of the system is the vector of mass positions, $s=(\vec{r}_1, \vec{r}_2, \ldots, \vec{r}_N)$, and the total elastic energy stored in the springs of the network is
\begin{equation}
    \mathcal{E}_{\rm elastic}(\vec{r}_1, \vec{r}_2, \ldots, \vec{r}_N) = \sum_{1 \leq i,j \leq N} \frac{1}{2} k_{ij} \left( \| \vec{r}_i - \vec{r}_j \| - \ell_{ij} \right)^2.
\end{equation}
Using the spring constants and spring rest lengths as trainable weights of the system, the partial derivatives of the energy function with respect to these weights are given by
\begin{equation}
    \frac{\partial \mathcal{E}_{\rm elastic}}{\partial k_{ij}} =  \frac{1}{2} \left( \| \vec{r}_i - \vec{r}_j \| - \ell_{ij} \right)^2, \qquad \frac{\partial \mathcal{E}_{\rm elastic}}{\partial \ell_{ij}} =  k_{ij} \left( \ell_{ij} - \| \vec{r}_i - \vec{r}_j \| \right).
\end{equation}
An experimental realization of an elastic network that learns using a variant of EP was performed by \citet{altman2023experimental}. In their implementation, they used the spring rest lengths $\ell_{ij}$ as trainable weights, while keeping the spring constants fixed (untrained).

\section{Concepts of quantum mechanics}
\label{sec:quantum-mechanics}

This section presents the concepts of quantum mechanics used in Quantum Equilibrium Propagation (QEP). Readers familiar with quantum mechanics may skip this section.

\paragraph{Quantum state.}
A quantum system is described by its state vector, denoted $| \psi \rangle$, which belongs to a complex vector space $\mathcal{H}$ equipped with an inner product $\langle \cdot| \cdot \rangle$ (specifically, a Hilbert space). For simplicity, we consider the case where $\mathcal{H}$ is finite-dimensional with dimension $d$.

\paragraph{Hamiltonian, eigenstates and energy levels.}
A central element of a quantum system is its Hamiltonian $\widehat{H}$, which is a linear operator acting on the Hilbert space, $\widehat{H} : \mathcal{H} \to \mathcal{H}$, with the property of being self-adjoint (i.e. $\widehat{H}$ equals its own adjoint). Because $\widehat{H}$ is self-adjoint, its eigenvalues are real. We denote the eigenvectors of $\widehat{H}$ as $|\psi_0 \rangle$, $|\psi_1 \rangle$, ..., $|\psi_{d-1} \rangle$, and the associated eigenvalues as $E_0 \leq E_1 \leq \ldots \leq E_{d-1}$, such that:
\begin{equation}
\label{eq:time-independent-schrodinger-equation}
\widehat{H} |\psi_k \rangle = E_k |\psi_k \rangle, \qquad 0 \leq k \leq d-1.
\end{equation}
The eigenvectors $|\psi_k \rangle$ are also called the eigenstates of the Hamiltonian, and their associated eigenvalues $E_k$ are the energy levels. The eigenstate $|\psi_0 \rangle$ with the lowest energy level is called the ground state. Eq.~\eqref{eq:time-independent-schrodinger-equation} is known as the time-independent Schrödinger equation.

\paragraph{Observable, measurement and expectation value.}
An observable represents a measurable physical quantity. In contrast with classical mechanics where an observable is a real-valued function of the system's state, in quantum mechanics, the mathematical representation of an observable is a self-adjoint operator, $\widehat{O} : \mathcal{H} \to \mathcal{H}$. The set of possible outcomes of measuring $\widehat{O}$ in state $|\psi\rangle$ is the set of eigenvalues of $\widehat{O}$, denoted $o_0$, $o_1$, ..., $o_{d-1}$, which are real due to the self-adjoint property.\footnote{The Hamiltonian $\widehat{H}$ is an example of an observable, with possible measurement outcomes being the energy levels $E_0$, $E_1$, ..., $E_{d-1}$.} A peculiar aspect of quantum mechanics is that measurement outcomes are inherently probabilistic. When the system is in state $|\psi\rangle$, the probability of obtaining outcome $o_k$ upon measuring $\widehat{O}$ is given by the Born rule, $p_k = |\langle o_k|\psi\rangle|^2$, where $|o_k \rangle$ is the eigenstate associated with $o_k$, i.e. such that $\widehat{O} |o_k \rangle = o_k |o_k \rangle$. The expectation value of a measurement of $\widehat{O}$ when the system is in state $|\psi\rangle$ is denoted $\langle \widehat{O} \rangle_{\psi}$ and calculated as $\langle\widehat{O} \rangle_{\psi} = \sum_{k=1}^d p_k o_k$. Using the spectral theorem for self-adjoint operators, it can be shown that this expectation value rewrites more concisely as
\begin{equation}
\langle \widehat{O} \rangle_{\psi} = \langle \psi | \widehat{O} |\psi \rangle.
\end{equation}
In statistical terms, this expectation value represents the average result of a large number of measurements of the observable $\widehat{O}$ performed on the system in state $|\psi\rangle$.


\paragraph{State collapse.} Another peculiar aspect of quantum mechanics is that a measurement on a quantum system inevitably changes its state (in contrast with classical mechanics where measurements on a system can be performed without significantly perturbing its state). Specifically, upon measurement of the observable $\widehat{O}$, if the measurement outcome is eigenvalue $o_k$, then the state $|\psi\rangle$ of the system instantaneously ``collapses'' to the corresponding eigenstate $|o_k\rangle$. This principle is known as state collapse, or state reduction. Consequently, if the quantum system is in state $|o_k\rangle$ after measuring $\widehat{O}$, then performing a second measurement of $\widehat{O}$ immediately after the first one will necessarily yield the same outcome $o_k$ and leave the state unchanged, in accordance with the Born rule ($p_k = \langle o_k | o_k \rangle = 1$). However, state collapse also has unusual consequences that do not occur in classical mechanics. To illustrate, suppose we want to measure two observables $\widehat{O}$ and $\widehat{P}$ in state $|\psi\rangle$, aiming to obtain (unbiased) estimates of both $\langle \widehat{O} \rangle_{\psi}$ and $\langle \widehat{P} \rangle_{\psi}$. The difficulty is that the system is no longer in state $|\psi\rangle$ after measuring $\widehat{O}$. In general, to obtain an (unbiased) estimate of $\langle \widehat{P} \rangle_{\psi}$, we must first reset the system to state $|\psi\rangle$ before measuring $\widehat{P}$.

\paragraph{Commuting observables.} One notable case where it is legitimate to measure $\widehat{O}$ and $\widehat{P}$ successively without resetting the state of the system between the two measurements is when the two observables commute, i.e.
\begin{equation}
    \widehat{O} \widehat{P} = \widehat{P} \widehat{O}.
\end{equation}
In this case, the two operators $\widehat{P}$ and $\widehat{O}$ are simultaneous diagonalizable: the eigenstates $|o_1 \rangle$, ..., $|o_{d-1} \rangle$ of $\widehat{O}$ are also eigenstates of $\widehat{P}$. Therefore, the probability of collapsing to any eigenstate $|o_k\rangle$ is the same under each observable (given by the Born rule, $p_k = |\langle o_k | \psi \rangle |^2$), and subsequent measurements of either $\widehat{O}$ or $\widehat{P}$ will leave the state unchanged. This allows for successive measurements of $\widehat{O}$ and $\widehat{P}$ to obtain unbiased estimates of their expectation values in the initial state $|\psi\rangle$, without resetting the system between the measurements.\footnote{In this derivation, we have assumed for simplicity that all eigenvalues of the operators are distinct, meaning the eigenspaces are one-dimensional. Although this reasoning does not apply when some eigenspaces have higher dimensions, the result still holds.}

\paragraph{Variational formulation of the Hamiltonian's eigenstates.}
The central result that allows us to transpose EP to quantum systems is the following variational formulation of the eigenstates of the system's Hamiltonian.

\begin{restatable}[Variational formulation of the ground state]{lma}{lmavariational}
\label{lma:variational-formulation}
The ground state $|\psi_0\rangle$ achieves the minimum of the Hamiltonian's expectation value:
\begin{equation}
|\psi_0\rangle = \underset{\psi \in \mathcal{H}, \|\psi\|=1}{\arg \min} \; \langle \psi | \widehat{H} | \psi \rangle.
\end{equation}
\end{restatable}

More generally, the eigenstates of the Hamiltonian $\widehat{H}$ are the critical points of the Rayleigh quotient $\psi \mapsto \frac{\langle \psi | \widehat{H} | \psi \rangle}{\langle \psi | \psi \rangle}$.

\section{Quantum equilibrium propagation}
\label{quantum-equilibrium-propagation}

Now we turn to Quantum Equilibrium Propagation (QEP), a version of EP for quantum systems. In this context, where we use the quantum system as a `learning machine', the Hamiltonian is parameterized by trainable weights $w=(w_1,w_2,\ldots,w_M)$, and an input $x$ can be supplied as a boundary condition to the system. We denote the corresponding Hamiltonian as $\widehat{H}(w,x)$. Next, we assume that the system reaches its ground state $|\psi(w,x) \rangle$, characterized by
\begin{equation}
\widehat{H} (w,x) |\psi(w,x) \rangle = E(w,x) |\psi(w,x) \rangle,
\end{equation}
where $E(w,x)$ is the ground state energy level, i.e. the lowest eigenvalue of $\widehat{H}(w,x)$. Similar to the classical setting, the ground state $|\psi(w,x) \rangle$ is used to encode a prediction on the desired output $y$, based on the supplied input $x$. We also assume an observable $\widehat{C}(y)$ parameterized by $y$, whose expectation value at the Hamiltonian's ground state represents the cost function to minimize,
\begin{equation}
\langle \widehat{C}(y) \rangle_{\psi(w,x)} = \langle \psi(w,x) | \widehat{C}(y) | \psi(w,x) \rangle.
\end{equation}
The goal is to adjust the trainable weights of the Hamiltonian to minimize this cost function. Similar to classical EP, we assume that $\widehat{C}(y)$ is the Hamiltonian of an interaction between the system's state $|\psi\rangle$ and desired output $y$, and that this interaction can be integrated into the system. Specifically, we form the `total Hamiltonian'
\begin{equation}
\widehat{H}^\beta = \widehat{H}(w,x) + \beta \widehat{C}(y),
\end{equation}
where $\beta \in \mathbb{R}$ controls the strength of the new interaction. Given an input-output pair $(x,y)$, QEP proceeds as follows.
\begin{enumerate}
\item Set $\beta=0$. For each $k \in \{ 1,2,\cdots,M \}$, repeat the following step $T$ times: reach the ground state $|\psi_\star^0\rangle$ of $\widehat{H}^0$, characterized by
\begin{equation}
\widehat{H}^0 |\psi_\star^0\rangle = E_\star^0 |\psi_\star^0\rangle,
\end{equation}
where $E_\star^0$ is the associated ground state energy, and measure the observable $\frac{\partial \widehat{H}^0}{\partial w_k}$. Denote the outcomes of the $T$ measurements $h_k^{(1)}(0), h_k^{(2)}(0), \ldots, h_k^{(T)}(0)$.
\item Set $\beta>0$ and proceed as above. For each $k \in \{ 1,2,\cdots,M \}$, repeat the following step $T$ times: reach the ground state $|\psi_\star^\beta\rangle$ of $\widehat{H}^\beta$, characterized by
\begin{equation}
\widehat{H}^\beta |\psi_\star^\beta\rangle = E_\star^\beta |\psi_\star^\beta\rangle,
\end{equation}
where $E_\star^\beta$ is the associated ground state energy, and measure the observable $\frac{\partial \widehat{H}^\beta}{\partial w_k}$. Denote the outcomes of the $T$ measurements $h_k^{(1)}(\beta), h_k^{(2)}(\beta), \ldots h_k^{(T)}(\beta)$.
\item Update the trainable weights $w_1, w_2, \ldots, w_M$ as
\begin{equation}
\label{eq:learning-rule-qep}
\Delta w_k = \frac{\eta}{\beta} \left[ \frac{1}{T} \sum_{t=1}^T h_k^{(t)}(0) - \frac{1}{T} \sum_{t=1}^T h_k^{(t)}(\beta) \right].
\end{equation}
\end{enumerate}

Similar to the classical case, the learning rule of QEP (Eq.~\eqref{eq:learning-rule-qep}) approximates one step of gradient descent on the expectation value of the cost observable.

\begin{restatable}[Quantum Equilibrium Propagation]{thm}{thmqep}
\label{thm:quantum-equilibrium-propagation}
The gradient of the cost function can be approximated as
\begin{align}
\label{eq:qep}
\nabla_w \langle \psi(w,x) | \widehat{C}(y) | \psi(w,x) \rangle & = \left. \frac{d}{d\beta} \langle \psi_\star^\beta | \frac{\partial \widehat{H}^\beta}{\partial w} | \psi_\star^\beta \rangle \right|_{\beta=0} \\
\label{eq:qep-2}
& \approx \frac{1}{\beta} \left[ \langle \psi_\star^\beta | \frac{\partial \widehat{H}^\beta}{\partial w} | \psi_\star^\beta \rangle - \langle \psi_\star^0 | \frac{\partial \widehat{H}^0}{\partial w} | \psi_\star^0 \rangle \right] \\
& \approx \frac{1}{\beta} \left[ \frac{1}{T} \sum_{t=1}^T h_k^{(t)}(\beta) - \frac{1}{T} \sum_{t=1}^T h_k^{(t)}(0) \right].
\end{align}
\end{restatable}

Theorem~\ref{thm:quantum-equilibrium-propagation} is proved in Appendix~\ref{sec:proof}. Here are a few important remarks.


\paragraph{Quantum measurements.}
Measurements in QEP differ from measurements in classical EP in several ways. First, since a measurement in quantum mechanics only gives an unbiased estimate of the expectation value, multiple measurements of the Hamiltonian derivatives $\frac{\partial \widehat{H}}{\partial w_k}$ are required to get better estimates of these expectation values. Second, since the state of the system generally changes upon measurement of an observable (Hamiltonian derivative), the system must be reset to its ground state after each measurement. Third, the observables $\frac{\partial \widehat{H}}{\partial w_j}$ and $\frac{\partial \widehat{H}}{\partial w_k}$ cannot generally be measured simultaneously unless they commute.

\paragraph{Gradient approximation.}
QEP involves two levels of approximation in the estimate of the gradient of the cost function. The first level is due to the finite difference used to approximate the derivative $\frac{d}{d\beta}$ at $\beta=0$ (also present in classical EP). The second is due to the probabilistic nature of quantum measurements. However, the insights of \citet{scellier2023energy} also apply in QEP: the learning rule of Eq.~\eqref{eq:learning-rule-qep} provides an unbiased estimator of the gradient of the contrastive function
\begin{equation}
\frac{1}{\beta} \left[ \langle \psi_\star^\beta | \widehat{H}^\beta | \psi_\star^\beta \rangle - \langle \psi_\star^0 | \widehat{H}^0 | \psi_\star^0 \rangle \right].
\end{equation}
As in the classical setting, this contrastive function is a lower bound of the true cost function if $\beta>0$ and an upper bound if $\beta<0$. The variant of \citet{laborieux2021scaling}, which combines $\beta<0$ and $\beta>0$, can also be used in QEP to improve the gradient estimator of the cost function. 

\paragraph{Using any eigenstate of the Hamiltonian.}
Similar to the classical setting where the equilibrium state need not be stable but only a critical point (stationary state) of the energy function, QEP applies to any eigenstate of the system's Hamiltonian, not just the ground state. In principle, however, one requirement for Eq.~\eqref{eq:qep-2} to hold is that the nudge eigenstate $|\psi_\star^\beta\rangle$ must be obtained as a smooth deformation (adiabatic transformation) of the free eigenstate $|\psi_\star^0\rangle$ when varying the nudging parameter from $0$ to $\beta \neq 0$.

\paragraph{Local learning rule.}
Similar to the classical setting, if the total Hamiltonian of the system can be expressed as the sum of Hamiltonians corresponding to individual interactions or contributions, i.e. $\widehat{H} = \widehat{H}_1 + \widehat{H}_2 + \cdots + \widehat{H}_M$ where $\widehat{H}_k$ is the Hamiltonian of an interaction parameterized by $w_k$, then the Hamiltonian derivatives simplify to $\frac{\partial \widehat{H}}{\partial w_k} = \frac{\partial \widehat{H}_k}{\partial w_k}$. If the trainable weight $w_k$ is stored close to where the observable $\frac{\partial \widehat{H}_k}{\partial w_k}$ is measured, the learning rule for $w_k$ becomes local.
Next we present for illustration the setting of the transverse-field Ising model and quantum harmonic oscillator network.

\subsection{Transverse-field Ising model}
\label{sec:quantum-ising-model}

As a first example of a quantum system that may be trained using QEP, we consider the transverse-field Ising model, a quantum version of the classical Ising model described in Section~\ref{sec:ising-model}. Recall that a classical Ising network of $N$ classical spins has $2^N$ possible configurations: each of the $N$ spins can be either `up' or `down'. In the quantum setting, a system of $N$ spins exists in a superposition of these $2^N$ configurations (i.e. a linear combination with complex coefficients). Hence the critical difference between the classical and the quantum settings: while the state of the classical model is described by a $N$-dimensional binary-valued vector, the state of the quantum model is described by a $2^N$-dimensional complex-valued vector. We denote the $d=2^N$ basis states as $| \sigma_1 \sigma_2 \cdots \sigma_N \rangle$ with $\sigma_k \in \{ \uparrow, \downarrow \}$ for each $k \in \{ 1,2,\ldots,N \}$, e.g. $| \uparrow \uparrow \cdots \uparrow \uparrow \rangle$, $| \uparrow \uparrow \cdots \uparrow \downarrow \rangle$ and similarly for the other $2^N-2$ basis states.

Similar to the Ising energy function~\eqref{eq:ising-energy-function}, the Hamiltonian of the transverse-field Ising model has couplings between spins $J_{jk} \in \mathbb{R}$ and bias fields $h_k \in \mathbb{R}$ applied to individual spins. These are the trainable weights. The Ising Hamitonian takes the form
\begin{equation}
\widehat{H}_{\rm Ising} = - \sum_{1 \leq j < k \leq N} J_{jk} \widehat{Z}_j \widehat{Z}_k - \sum_{k=1}^N h_k \widehat{X}_k,
\end{equation}
where $\widehat{Z}_k$ and $\widehat{X}_k$ are the Pauli operators, defined as follows. The Pauli $\widehat{Z}_k$ operator acts as a phase-flip operator on the $k$-th spin, according to:
\begin{align}
\widehat{Z}_k \; | \sigma_1 \cdots \sigma_{k-1} \uparrow \sigma_{k+1} \cdots \sigma_N \rangle & = + | \sigma_1 \cdots \sigma_{k-1} \uparrow \sigma_{k+1} \cdots \sigma_N \rangle, \\
\widehat{Z}_k \; | \sigma_1 \cdots \sigma_{k-1} \downarrow \sigma_{k+1} \cdots \sigma_N \rangle & = - | \sigma_1 \cdots \sigma_{k-1} \downarrow \sigma_{k+1} \cdots \sigma_N \rangle.
\end{align}
The Pauli $\widehat{X}_k$ operator acts as a bit-flip operator on the $k$-th spin, according to:
\begin{align}
\widehat{X}_k \; | \sigma_1 \cdots \sigma_{k-1} \uparrow \sigma_{k+1} \cdots \sigma_N \rangle & = | \sigma_1 \cdots \sigma_{k-1} \downarrow \sigma_{k+1} \cdots \sigma_N \rangle, \\
\widehat{X}_k \; | \sigma_1 \cdots \sigma_{k-1} \downarrow \sigma_{k+1} \cdots \sigma_N \rangle & = | \sigma_1 \cdots \sigma_{k-1} \uparrow \sigma_{k+1} \cdots \sigma_N \rangle.
\end{align}
In this setting, the gradients of the Ising Hamiltonian with respect to the trainable weights, required in the learning rule of Eq.~\eqref{eq:learning-rule-qep}, are given by
\begin{equation}
\frac{\partial \widehat{H}_{\rm Ising}}{\partial J_{jk}} = - \widehat{Z}_j \widehat{Z}_k, \qquad \frac{\partial \widehat{H}_{\rm Ising}}{\partial h_k} = - \widehat{X}_k.
\end{equation}
Importantly, the Pauli $\widehat{Z}_k$ operators ($1 \leq k \leq N$) commute with one another, allowing them to be measures simultaneously (when the system is in the ground state). Similarly, the Pauli $\widehat{X}_k$ operators ($1 \leq k \leq N$) also commute and can be measured simultaneously. However, $\widehat{Z}_j$ and $\widehat{X}_k$ do not commute, so they cannot be measured simultaneously.

\subsection{Quantum harmonic oscillator network}

The quantum harmonic oscillator network is a quantum analogue of the elastic network model presented in Section~\ref{sec:elastic-network}. It consists of $N$ quantum particles, such as atoms, interacting via harmonic potentials. Let's assume for simplicity that each atom is described by its 1D position (rather than 3D position). In classical mechanics, the state of such a system would be described by the $N$-dimensional vector of positions of the atoms $(r_1,r_2,\cdots,r_N) \in \mathbb{R}^N$. In contrast, in quantum mechanics, the state of the system is a superposition of all these configurations, described by a wave function $\psi : \mathbb{R}^N \to \mathbb{C}$ that assigns a complex number $\psi(r_1, r_2, \cdots, r_N)$ to each possible configuration $(r_1, r_2, \cdots, r_N)$.

The position and momentum operators of the $i$-th atom, denoted as $\widehat{r}_i$ and $\widehat{p}_i$, are defined by their action on the wavefunction as follows:
\begin{align}
    (\widehat{r}_i \psi) (r_1,r_2,\cdots,r_N) & = r_i \psi (r_1,r_2,\cdots,r_N), \\
    (\widehat{p}_i \psi) (r_1,r_2,\cdots,r_N) & = -\textbf{i} \hbar \frac{\partial \psi}{\partial r_i} (r_1,r_2,\cdots,r_N),
\end{align}
where $\textbf{i}$ is the imaginary unit ($\textbf{i}^2 = -1$) and $\hbar$ is the reduced Planck constant. Similar to the elastic energy function, the Hamiltonian of the quantum harmonic oscillator network is given by:
\begin{equation}
    \widehat{H}_{\rm QHO} = \sum_{i=1}^N \frac{\widehat{p}_i^2}{2 m_i} + \frac{1}{2} \sum_{1 \leq i,j \leq N} k_{ij} \left( \widehat{r}_i - \widehat{r}_j \right)^2,
\end{equation}
where $\frac{\widehat{p}_i^2}{2 m_i}$ is the kinetic energy operator of the $i$-th atom, and $\frac{1}{2} k_{ij} \left( \widehat{r}_i - \widehat{r}_j \right)^2$ is the harmonic potential operator between the $i$-th and $j$-th atoms. In these expressions, $m_i$ are the masses of the atoms and $k_{ij}$ are the spring constants, which serve as trainable weights. The partial derivatives of the Hamiltonian with respect to the spring constants are given by
\begin{equation}
    \frac{\partial \widehat{H}_{\rm QHO}}{\partial k_{ij}} =  \frac{1}{2} \left( \widehat{r}_i - \widehat{r}_j \right)^2.
\end{equation}
It is straightforward to verify that all $\widehat{r}_i$ operators commute with each other. This allows us to measure these observables simultaneously (to obtain the gradients of the cost function).

\section{Discussion}

Equilibrium Propagation (EP) has been studied in various classical physical systems, and has been implemented experimentally in resistor networks, classical Ising networks and elastic networks. EP is generally applicable to systems that extremize an energy functional. Quantum Equilibrium Propagation (QEP) extends EP to quantum systems, where the extremized functional is the mean energy, achieving its extrema at the eigenstates of the system's Hamiltonian. QEP could serve both as a normative framework for training quantum systems, and to justify the use of EP in physical systems where quantum effects may arise and where the classical EP framework is not directly applicable. For instance, QEP could provide insights into the work by \citet{laydevant2024training}, which employed quantum annealing in the context of classical EP.

One attractive feature of QEP, similar to EP, is the locality of the learning rule. This locality might be advantageous for building specialized quantum computers with reduced classical computation overhead, where the physical quantities serving as trainable weights would be located near the location where the Hamiltonian derivatives are measured. A complication in the quantum setting is the probabilistic nature of quantum measurements. Multiple measurements are often required to obtain accurate gradient estimates of the cost function, which could necessitate additional memory to store the outcomes. To address this, studying the effect of single measurements of single eigenstates (free or nudge) on the variance of the gradient estimator would be useful, similar to the study by \citet{williams2023flexible} in the classical setting. Another difference with the classical setting is that Hamiltonian derivatives generally cannot be measured simultaneously, unless they commute. We discussed the transverse-field Ising model and the quantum harmonic oscillator network (quantum analogues of the classical Ising network and elastic network) and we have seen that measurements in these models can largely be parallelized.

Similar to the classical setting where EP can be applied to any stationary state (critical point) of the system's energy function, QEP can be applied to any eigenstate of the system's Hamiltonian, not just the ground state. While this feature provides greater flexibility, a caveat is that for Eq.~\eqref{eq:qep-2} to hold, the nudge eigenstate must in principle be obtained as a smooth deformation (adiabatic transformation) of the free eigenstate. It remains to be seen whether this condition is necessary or can be further relaxed in practice.

The next step would be to simulate QEP numerically. For quantum Ising networks, the exact diagonalization method becomes impractical when the number of spins $N$ exceeds a few dozen due to the exponential growth of the state space ($d = 2^N$). For larger systems, approximate methods such as the Density Matrix Renormalization Group (DMRG) and Variational Monte Carlo (VMC) could be employed.

Finally, two recent theoretical advances in EP might also be useful in the quantum setting to expand QEP. The first is  `Holomorphic EP' (HEP), introduced by \citet{laborieux2022holomorphic}. HEP is a version of EP that applies to systems described by complex-valued state variables governed by a holomorphic energy function. HEP's advantage is that it allows for extracting the exact weight gradient of the cost function, rather than an approximation of it. To achieve this, HEP uses a complex nudge $\beta \in \mathbb{C}$ and applies the Cauchy formula of complex analysis to Eq.~\eqref{eq:ep-formula}, rather than the finite difference method. Applying HEP to quantum systems would be interesting, but one difficulty is that the total Hamiltonian $\widehat{H} + \beta \widehat{C}$ is not self-adjoint when $\beta \in \mathbb{C} \setminus \mathbb{R}$, and is therefore incompatible with standard formulations of quantum mechanics. The second theoretical advance is `Agnostic EP' (AEP), introduced in \citet{scellier2022agnostic}. In AEP, weight updates are performed through physical dynamics, without relying on measurements and feedback. AEP thus overcomes several challenges related to implementing the contrastive learning rules of Eq.~\ref{eq:learning-rule-ep} and Eq.~\ref{eq:learning-rule-qep} (discussed above). To achieve this, in AEP, trainable weights are seen as variables that minimize the system's energy function, while `control variables' are coupled with the trainable weights and can perform homeostatic control on them. The term `Agnostic' reflects the minimal analytical knowledge required about the system, unlike the contrastive learning rule that necessitates knowledge of the partial derivatives of the energy function with respect to the trainable weights.

\begin{ack}
The author thanks Yassir Akram, Nicolas Zucchet and Jo\~ao Sacramento for useful discussions about quantum mechanics and EP, as well as Suhas Kumar for useful feedback on the manuscript.
\end{ack}

\bibliographystyle{abbrvnat}
\bibliography{biblio}

\clearpage
\appendix

\newpage
\section{Proof of Theorem~\ref{thm:quantum-equilibrium-propagation}}
\label{sec:proof}

In this appendix, we prove Theorem~\ref{thm:quantum-equilibrium-propagation} using Lemma~\ref{lma:variational-formulation}. First we repeat and prove Lemma~\ref{lma:variational-formulation}.

\lmavariational*

\begin{proof}[Proof of Lemma~\ref{lma:variational-formulation}]
Let $|\psi_0 \rangle$, $|\psi_1 \rangle$, ..., $|\psi_{d-1} \rangle$ the eigenvectors of $\widehat{H}$, and $E_0 \leq E_1 \leq \ldots \leq E_{d-1}$ the associated eigenvalues, such that:
\begin{equation}
\widehat{H} |\psi_k \rangle = E_k |\psi_k \rangle, \qquad 0 \leq k \leq d-1.
\end{equation}
Due to the self-adjoint property of $\widehat{H}$, by the spectral theorem, the eigenstates of $\widehat{H}$ form an orthonormal basis of $\mathcal{H}$,
\begin{equation}
\langle \psi_j |\psi_k \rangle = \left\{
\begin{array}{l}
    1 \qquad \text{if } j=k, \\
   0 \qquad \text{if } j \neq k,
\end{array}
\right. \qquad 0 \leq j,k \leq d-1.
\end{equation}
In this orthonormal basis, the following decomposition holds for any state $|\psi\rangle$,
\begin{equation}
\forall |\psi\rangle \in \mathcal{H}, \qquad |\psi \rangle = \sum_{k=0}^{d-1} | \psi_k \rangle \langle \psi_k | \psi \rangle.
\end{equation}
It follows that
\begin{equation}
\langle \psi | \widehat{H} | \psi \rangle = \sum_{k=1}^d |\langle \psi_k | \psi \rangle|^2 E_k \qquad \text{and} \qquad \langle \psi | \psi \rangle = \sum_{k=0}^{d-1} |\langle \psi_k | \psi \rangle|^2,
\end{equation}
so that
\begin{equation}
\langle\psi|\widehat{H}|\psi\rangle = \sum_{k=0}^{d-1} |\langle \psi_k | \psi \rangle|^2 E_k \geq \sum_{k=0}^{d-1} |\langle \psi_k | \psi \rangle|^2 E_0 = E_0,
\end{equation}
where we have used that $|\langle \psi_k | \psi \rangle|^2 \geq 0$, $E_k \geq E_0$ for every $k$, and $\|\psi\|=1$. Equality holds when $|\langle \psi_k | \psi \rangle|^2 = 0$ for $k \geq 1$, i.e. for $| \psi \rangle = \alpha | \psi_0 \rangle$ with $\alpha \in \mathbb{C}$ and $|\alpha|=1$.
\end{proof}

Next we repeat and prove Theorem~\ref{thm:quantum-equilibrium-propagation}.

\thmqep*

\begin{proof}[Proof of Theorem~\ref{thm:quantum-equilibrium-propagation}]
Define the mean energy functional and mean cost functional
\begin{equation}
\forall \psi \in \mathcal{H}, \qquad \mathcal{E}_{\rm mean}(w,x,\psi) = \langle\psi|\widehat{H}(w,x)|\psi\rangle, \qquad \mathcal{C}_{\rm mean}(\psi,y) = \langle\psi|\widehat{C}(y)|\psi\rangle.
\end{equation}
The mean total energy functional associated to $\mathcal{E}_{\rm mean}$ and $\mathcal{C}_{\rm mean}$, as defined by Eq.~\eqref{eq:total-energy-function}, is
\begin{align}
\mathcal{E}_{\rm mean}^\beta(w,x,\psi,y) & = \mathcal{E}_{\rm mean}(w,x,\psi) + \beta \mathcal{C}_{\rm mean}(\psi,y) \\
& = \langle \psi | \widehat{H}(w,x) | \psi \rangle + \beta \langle \psi | \widehat{C}(y) | \psi \rangle \\
& = \langle \psi | \left[ \widehat{H}(w,x) + \beta \widehat{C}(y) \right] | \psi \rangle = \langle \psi | \widehat{H}^\beta | \psi \rangle.
\end{align}
By Lemma~\ref{lma:variational-formulation}, we have
\begin{equation}
\psi_\star^\beta = \underset{\psi \in \mathcal{H}, \| \psi \|=1}{\arg \min} \; \mathcal{E}_{\rm mean}^\beta(w,x,\psi,y) \qquad \Longleftrightarrow \qquad \widehat{H}^\beta | \psi_\star^\beta \rangle = E_\star^\beta | \psi_\star^\beta \rangle,
\end{equation}
which allows us to apply Theorem~\ref{thm:equilibrium-propagation} in the QEP algorithm:
\begin{equation}
\nabla_w \mathcal{C}_{\rm mean}(\psi(w,x),y) = \left. \frac{d}{d\beta} \frac{\partial \mathcal{E}_{\rm mean}^\beta}{\partial w}(w,x,\psi_\star^\beta,y) \right|_{\beta=0}.
\end{equation}
Using the explicit forms of the `mean cost functional' and `mean total energy functional', we obtain
\begin{equation}
\nabla_w \langle \psi(w,x) | \widehat{C}(y) | \psi(w,x) \rangle = \left. \frac{d}{d\beta} \langle \psi_\star^\beta | \frac{\partial \widehat{H}^\beta}{\partial w} | \psi_\star^\beta \rangle \right|_{\beta=0},
\end{equation}
which is Eq.~\eqref{eq:qep}.
\end{proof}

\end{document}